\documentclass{article}
\usepackage{booktabs} 

\usepackage{fullpage}   
\usepackage[utf8]{inputenc}
\usepackage{amsmath, amsthm, hyperref}
\usepackage{amsfonts}
\usepackage{algorithm}
\usepackage{graphicx}
\usepackage{cleveref}
\usepackage{xspace}
\usepackage{thm-restate}
\usepackage{url}
\usepackage{algorithm}
\usepackage{algpseudocode}
\usepackage{color,authblk}

\newcommand{\eps}{\varepsilon}
\newcommand{\laplacenoise}{\text{Lap}}

\newcommand{\calA}{\mathcal{A}}

\newcommand{\cost}{\text{cost}}
\newcommand{\diam}{\text{diam}}
\newcommand{\dist}{\text{dist}}
\newcommand{\globalS}{\text{OPT}}
\newcommand{\opt}{\text{OPT}}

\newcommand{\calD}{\mathcal{T}}
\newcommand{\badcut}{\alpha_C}
\newcommand{\badcutF}{\alpha_F}

\newcommand{\E}{\mathbb{E}}
\newcommand{\R}{\mathbb{R}}
\newcommand{\calB}{\mathcal{B}}
\newcommand{\bcc}{\text{ badly-cut}}
\newcommand{\poly}{\text{poly}}

\newcommand{\skin}{\texttt{SKYNTYPE}\xspace}
\newcommand{\shuttle}{\texttt{SHUTTLE}\xspace}
\newcommand{\higgs}{\texttt{HIGGS}\xspace}
\newcommand{\synth}{\texttt{SYNTHETIC}\xspace}
\newcommand{\covtype}{\texttt{COVERTYPE}\xspace}

\newtheorem{lemma}{Lemma}

\newtheorem{theorem}{Theorem}

\newtheorem{fact}{Fact}

\title{Scalable Differentially Private Clustering via Hierarchically Separated Trees}

\author[1]{Vincent Cohen-Addad\footnote{All authors contributed equally to this work.}}
\author[1]{Alessandro Epasto} 
\author[1]{Silvio Lattanzi} 
\author[1]{Vahab Mirrokni}
\author[1]{Andres Munoz}
\author[2]{David Saulpic}
\author[3]{Chris Schwiegelshohn}
\author[1]{Sergei Vassilvitskii}
\affil[1]{Google Research}
\affil[2]{Sorbonne Université, LIP6, France}
\affil[3]{Aarhus University, Denmark}
\date{}

\begin{document}
\maketitle

\begin{abstract}
We study the private $k$-median and $k$-means clustering problem in $d$ dimensional Euclidean space. 
By leveraging tree embeddings, we give an efficient and easy to implement algorithm, that is empirically competitive with state of the art non private methods. 
We prove that our method computes a solution with cost at most $O(d^{3/2}\log n)\cdot OPT + O(k d^2 \log^2 n / \epsilon^2)$, where $\epsilon$ is the privacy guarantee. (The dimension term,  $d$, can be replaced with $O(\log k)$ using standard dimension reduction techniques.) Although the worst-case guarantee is worse than that of state of the art private clustering methods, the algorithm we propose is practical, runs in near-linear, $\tilde{O}(nkd)$, time and scales to tens of millions of points. We also show that our method is amenable to parallelization in large-scale distributed computing environments. In particular we show that our private algorithms can be implemented in logarithmic number of MPC rounds in the sublinear memory regime.
Finally, we complement our theoretical analysis with an empirical evaluation demonstrating the algorithm's efficiency and accuracy in comparison to other privacy clustering baselines.
\end{abstract}

\section{Introduction}
Clustering is a central problem in unsupervised learning with many applications
such as duplicate detection, community detection, computational
biology and many others. Several formulations of clustering problems
have been studied throughout the years. Among these,
the geometric versions of the problem have attracted a lot of
attention for their theoretical and practical importance.
In those problems we are given as input $n$ points and the objective is to 
put points that are close in the same cluster and far away points in different clusters.
Classic formulations of geometric clustering problem are $k$-means, 
$k$-median, and $k$-center.
Due to their relevance and their many practical applications, the 
problems have been extensively
studied and many algorithms~\cite{ahmadian2017better, arthur2007k,  byrka2014improved,  kanungo2004local, jain2003greedy, li20111} and
heuristics~\cite{lloyd1982least} have been proposed to solve
the classic version of these problems.

In this paper, we study these problems through a differential privacy lens. {\em Differential Privacy (DP)} has emerged as a de facto standard for capturing user privacy~\cite{dwork}. It is characterized by the notion of neighboring datasets, $X$ and $X'$, generally assumed to be differing on a single user's data.  An algorithm $\calA$ is $\epsilon$-differentially private if the probability of observing any particular outcome $S$ when run on $X$ vs $X'$ is bounded:
$Pr[\calA(X) = S] \leq$ $ e^{\epsilon} Pr[\calA(X') = S].$

For an introduction, see the book by Dwork and Roth~\cite{dwork}. At a high level, DP forces an algorithm to not focus on any individual training example, rather capturing global trends present in the data. 

\paragraph{Private Clustering}
Differentially private clustering has been well studied, with a number of works giving polynomial time approximately optimal algorithms for different versions of the problem, including  $k$-median and $k$-means~\cite{balcan,anamayclustering,badih_approximation,badih_local, clustering_with_convergence}. From an analysis standpoint, any approximately optimal differentially private algorithm must pay both a multiplicative as well as an additive approximation. In other words, the cost of any algorithm solution, $\textsc{Alg}$, will satisfy  $\textsc{Alg} \leq \alpha \textsc{OPT} + \beta$, for some $\alpha, \beta > 0$, where $\textsc{OPT}$ denotes the optimal solution cost. 

All else being equal, we aim for algorithms that minimize $\alpha$ and $\beta$, and recent work ~\cite{badih_approximation} has made a lot of progress in that direction. However, in a push to minimize $\alpha$ and $\beta$, algorithms often pay in added complexity and running time. In fact, all known differentially private clustering algorithms with theoretical guarantees have this shortcoming: they have superlinear running times and do not scale to large datasets, even though the large data regime is precisely the one for which using private methods is particularly important. Hence there is a big gap between optimal algorithms in theory and those that can be used in practice. 

\paragraph{Previous Work.} 
In terms of approximation guarantee, the result of Ghazi et al.~\cite{badih_approximation} is impressive: they show that it is possible to get privately the same approximation factor as the best non-private algorithm. This concludes a long line of work (see e.g. \cite{FeldmanFKN09, balcan, StemmerK18}) that focused on approximation guarantee, but not really on practical algorithms. Furthermore, those algorithm attempt to minimize the multiplicative approximation factor, dropping the additive term: instead, in the hope to improve the practical guarantee, Jones et al.~\cite{JonesNN21} and later Nguyen et al~\cite{anamayclustering} proposed an algorithm with (large) constant multiplicative error, but with additive error close to the optimal one ($O(k\log n)$ for $k$-median, $O(k\log n + k\sqrt d)$ for $k$-means). They implement their algorithm for $k$-means, showing guarantees comparable to \cite{balcan}, and quite far from the results of the non-private Lloyd's algorithm. Further, those algorithms all have super linear running time, and do not scale nicely to large datasets. 

On the opposite side, an algorithm for was recently described by Chang and Kamath.\footnote{\url{https://ai.googleblog.com/2021/10/practical-differentially-private.html}} This algorithm seems to perform extremely well in practice, with result close to the non-private $k$-means++, and can be implemented in a distributed environment to handle large scale dataset. However, this algorithm has no theoretical guarantee. 

\paragraph{Our Results and Techniques.} 
In this paper, we aim to address the above shortcomings, and design practical differentially private clustering algorithms, with provable approximation guarantees, that are fast, and amenable to scale up via parallel implementation. 
Toward this goal, we take an approach that has been successful in the non-private clustering literature. While there are constant-approximate algorithms for $k$-median and $k$-means, see \cite{ahmadian2017better,kanungo2004local, jain2003greedy}, most practical implementations use the $k$-means++ algorithm of~\cite{arthur2007k}, which has a $\Theta(\log k)$ approximation ratio. The reason for the success of $k$-means++ is two-fold. First, it is fast, running in linear time, and second, it performs well empirically despite the logarithmic worst-case guarantee. The methods we introduce in this work, while different from $k$-means++, have the same characteristics: they are fast, and perform {\em much} better than their worst-case guarantees, significantly outperforming all other implementations. In particular, they run in near-linear time, are amenable to parallel implementation in logarithmic number of rounds, and output high-quality private clusters in practice.

 Our first contribution is an efficient and scalable algorithm for differentially private $k$-median. Our starting point is an embedding of the input points into a tree using a randomly-shifted quadtree (sometimes called HST for Hierarchically Separated Tree)\footnote{We note here that this technique has already been used in prior work for private clustering~\cite{balcan} to find a  $\tilde O(n)$-sized set of candidate centers, to then run a polynomial time local search algorithm. In contrast, we use this structure to directly compute a solution.}. It is well known~\cite{FakcharoenpholRT03} that such tree embeddings can approximately preserve pairwise distances. Our key insight is that it is possible to truncate the tree embedding so that leaves represent sets of points of large enough cardinality and then use them to compute a solution for the $k$-median problem. In fact, by using this insight and by carefully adding Laplace noise to the cardinality of the sets considered by the algorithm, we obtain our differentially private $k$-median algorithm.  

Our second contribution is a parallel implementation of our algorithm in the classic massively parallel computing (MPC) model. This model is a theoretical abstraction of real-world systems like MapReduce \cite{dean2008mapreduce}, Hadoop \cite{white2012hadoop}, Spark \cite{zaharia2010spark} and Dryad \cite{isard2007dryad} and it is the standard for analyzing algorithms for large-scale parallel computing~\cite{karloff2010model, goodrich2011sorting, beame2013communication}. Interestingly we show that our algorithm can be efficiently implemented using a logarithmic number of MPC parallel rounds for $k$-median clustering. To the best of our knowledge, our algorithms are the first differentially private algorithms for $k$-median that can be efficiently parallelized.

Third, we complement our theoretical results with an in-depth experimental analysis of the performance of our $k$-median algorithm. We demonstrate that not only our algorithms scale to large datasets where, until now, differential private clustering remained elusive, but also we outperform other state-of-the-art private clustering baselines in medium-sized datasets. 
In particular, we show that in practice compared to prior work with theoretical guarantees such as~\cite{balcan}, our parallel algorithm can scale to more $40$ times larger datasets, improve the cost by up to a factor of $2$, and obtain solutions within small single digit constant factor of non-private baselines.

Finally, we adapt those techniques to the $k$-means problem. This poses an additional challenge because randomly-shifted quadtrees do not preserve squared distances accurately and so we cannot apply our approach directly. We adapt a technique first introduced by~\cite{CohenAddadFS19} to our setting. The key observation behind our approach is that even if randomly-shifted quadtrees do not preserve all the squared distances well, they accurately preserve most of them. 

Suppose that we are given in input some solution $S$ for our problem (later will clarify how to obtain it), then we show that for most centers in $S$ it is true that the distances to these centers are approximately preserved by the quadtree. So points living in the clusters of these centers can be clustered using  the tree embedding, and the remaining points can be clustered using the few centers that are not preserved in the solution $S$. 
Interestingly, we  prove that we can use this approach iteratively starting with a solution to the differentially private $1$-means problem as $S$. In Section~\ref{sec:kmeans} we show that this approach leads to an efficient differentially private algorithm for $k$-means.

\section{Preliminaries}\label{sec:prelim}
\paragraph{Notations.}
For two points $p$ and $q$ in $\mathbb{R}^d$, we let $\dist(p, q):=\|p-q\| = \sqrt{\sum_{i=1}^d (p_i-q_i)^2}$ be the Euclidean distance between $p$ and $q$. Given $r \ge 0$, we define $B(x,r) = \{y \in \Re^d | \; \dist(x, y) \le r\}$ as the closed ball around $x$ of radius $r$. 

We are given a set of points $P$ as input, and assume that $P$ is contained in the open ball $B(0, \Lambda)$.

We seek to find $k$ centers $C = \{c_1, \ldots c_k\}$, that approximately minimize the $(k,z)$-clustering where distances of every point to their closest center are raised to the power of $z\geq 1$
$$
\cost(P, C) = \sum_{p \in P} \min_{c \in C} \dist(p, c)^z \mathrm{.}
$$
We use $\opt$ to indicate an optimal solution to the problem. We define $\dist(p, C)  = \min_{c \in C} \dist(p, c)$. In this article, we focus on $k$-median ($z = 1$) and $k$-means ($z = 2$).

We say that a solution $C$ is $(\alpha,\beta)$-approximate if the $\cost(P,C) \le \alpha \cost(P, \opt) + \beta$. We will seek solutions where $\alpha$ is $\tilde O (\poly \log(k))$ and $\beta$ is $\tilde O (\poly(k, d, \log(n))\Lambda^z)$.

The goal of this paper is to have private algorithms that are easy to parallelize and that run in {\it near-linear} running time, where by near linear we mean time $\tilde O(n \cdot \poly(\log(n),k,d,\log(\Lambda)))$. Notice that celebrated k-means++~\cite{arthur2007k} satisfies all the requirements
with its $O(nkd)$ running time and $O(\log(k))$ approximation, except that it is not private. 

\paragraph{Differential privacy.}
\
We will make use of standard composition properties of differentially private algorithms, described in \cite{dwork}. The algorithm $\mathcal{A}$ that applies successively two algorithm $\mathcal{A}_1$ and $\mathcal{A}_2$ that are respectively $\eps_1$-DP and $\eps_2$-DP is itself $(\eps_1 + \eps_2)$-DP. If the $\mathcal{A}_1$ and $\mathcal{A}_2$ run on two distinct parts of the dataset, then $\mathcal{A}$ is $\max(\eps_1, \eps_2)$-DP. 

Lastly, if $\mathcal{A} : D_1 \times D_2 \rightarrow Z$ satisfies that for all $X \in D_1$, the algorithm $ \mathcal{A}(X, \cdot)$ is $\eps_1$-DP, and some algorithm
$\mathcal{B} : D_2 \rightarrow D_1$ is $\eps_2$-DP, then the algorithm $X \rightarrow \mathcal{A}(\mathcal{B}(X), X)$ is $(\eps_1 + \eps_2)$-DP. 

A standard differentially private algorithm is the Laplace Mechanism (see \cite{dwork}). We say that a random variable follows distribution $\laplacenoise(b)$ if its probability density function is $\frac{1}{2b}\exp\left(-\frac{|x|}{b}\right)$. With a slight abuse of notation, we use $\laplacenoise(b)$ to denote a variable that follows such a distribution. Example 3.1 in \cite{dwork} shows that the following algorithm for counting queries is $\eps$-DP :
$\mathcal{A}(X) = |X| + \laplacenoise(1/\eps)$.

Note that the notion of differential privacy is only a model, and our result should not be used blindly to preserve privacy of users.
We emphasize in particular that the privacy notion is with respect to a single user's data: hence, this model does not necessarily ensure privacy for a \emph{group} of people.

\paragraph{Randomly-shifted Quadtrees.}
A quadtree is a binary tree $T$, such that each node $x$ in the tree corresponds to a region $T(x)$ of $\mathbb{R}^d$. To distinguish from the input, we call tree nodes \emph{cells}. Each cell is a hyper-rectangle. For a cell $c$ with children $c_1, c_2$, the region spanned by $c$ is the union of those spanned by $c_1$ and $c_2$, i.e., $T(c) = T(c_1) \cup T(c_2)$. 

A shifted quadtree is constructed as follows. Start from a root cell containing the entire $d$-dimensional hypercube $[-\Lambda, \Lambda]^d$ at depth $0$, and proceed recursively. 
Let $c$ be a cell at depth $d\cdot i + j$, with $0 \leq j < d$. The $j$-th coordinate of the region spanned by $c$ is comprised in $[m, M]$. The children of $c$ are constructed as follows: let $x$ be some random number in $[m + \frac{M-m}{3}, M - \frac{M-m}{3}]$. $c_1$ comprises all points of $c$ that have their $j$-th coordinate at most $x$, and $c_2$ the remaining points\footnote{Another standard way of defining quadtree is to have $2^d$-regular trees, and to split along the $d$-dimensions at each step. We are more comfortable working with binary trees, which allows for a simpler dynamic program.}. Note that the diameter of the cells is divided by at least $3/2$ every $d$ levels. Denote by $\diam(c)$ the diameter of cell $c$.

\begin{algorithm}\caption{DP-kMedian($P$)}
\label{alg:kmed}
\begin{algorithmic}[1]
\State Compute a shifted quadtree $T$. Let $r$ be the root of $T$.
\State $w = $ MakePrivate($T, P$). 
\State Compute $v, S = $DynamicProgram-kMedian($T, w, r$)
\State \textbf{Return} $S_k$
\end{algorithmic}
\end{algorithm}

\begin{algorithm}
\caption{MakePrivate($T, P$)}\label{alg:makePrivate}
\begin{algorithmic}[1]
\State Input: a quadtree $T$, a set of points $P$
\State let $Q$ be a queue, initially containing only the root of $T$.
\State Let $w : T \rightarrow \mathbb{N}$, initiated with $\forall c,~w(c) = 0$.
\While{$Q$ is not empty}
\State Let $c = Q.$pop()
\If{$\diam(c) > \Lambda/n$}
\State let $w(c) = |T(c) \cap P| + \laplacenoise(d \log n / \eps)$
\If{$w(c) > 2 d \log n / \eps$}
\State Add $c$'s children to $Q$
\EndIf
\EndIf
\EndWhile

\State \textbf{Return} $w$
\end{algorithmic}
\end{algorithm}

A quadtree $\calD$ induces a metric: for two points $p$, $q$, we define $\dist_\calD(p, q) = \diam(c)$ where $c$ is the smallest cell that contains both $p$ and $q$. We will frequently use that the expected distortion between two points $p$ and $q$ in a shifted quadtree of depth $d\cdot \alpha$, for any $\alpha$, is $\E_\calD[\dist_\calD(p, q)] \leq d^{3/2} \alpha \dist(p, q).$

\begin{lemma}\label{lem:ballCut}[Reformulation of Lemma 11.3 \cite{har2011geometric}]
  For any $i$, radius $r$ and point $p$, we have that $\Pr[B(p, r) \text{ is cut at level } i] = O\left(\frac{d r}{2^i}\right)$.
\end{lemma}

In our case, we stop the construction when reaching cells of diameter $\Lambda / n$. Hence, $\alpha = \log n$
and the expected distortion is $O(d^{3/2} \log n)$. 
Such trees are often called Hierarchically Separated Trees (HST) in the literature.

\paragraph{Dimension Reduction.} 
For clustering problems, it is possible to use the Johnson-Lindenstrauss Lemma to reduce the dimension to $O(\log k)$ (see \cite{makarychev2019performance}). Hence, we can apply all our algorithms in such a dimension, replacing dependency in $d$ by $\log k$. 
To compute centers in the original space, we can extract the clusters from the low-dimensional solution and compute privately the $1$-median (or $1$-mean) of the cluster in the original $d$-dimensional space. This adds an additive error $O(kd)$. Later we describe how to implement this procedure in MPC.

\section{Simple algorithm for $k$-Median}

\paragraph{Algorithm}

A simple way of solving $k$-median is to embed the input points into a tree metric. Tree metrics are sufficiently simple to admit a dynamic program for computing an optimum. The approximation factor of this algorithm is therefore the distortion incurred by the embedding. 
We adapt this approach to incorporate privacy as follows. First, we embed the input into a quadtree, which is a hierarchical decomposition of $\mathbb{R}^d$ via axis-aligned hyperplanes. We then add noise on the quadtree to enforce privacy.
Subsequently, we run the dynamic program on the quadtree.
Unfortunately, a naive implementation of the dynamic program falls short of the nearly linear time algorithm we are hoping for.
We speed this up by trimming the recursion of the dynamic program for quadtree cells containing few points. 
To do this, we require a private count of the number of points in each cell, that guides the dynamic program.
Such a private count can be obtained from the Laplace mechanism (see the preliminaries), and the error incurred by the privacy is charged to the additive term of the approximation. The result we aim to show is the following theorem:
\begin{theorem}\label{thm:kmed}
Algorithm~\ref{alg:kmed} is $\eps$-DP and computes a solution with expected cost at most 
$O(d^{3/2}\log n)\cdot \opt + \frac{d^2 \cdot \log^2 n \cdot k}{\eps}\cdot \Lambda$. Furthermore, it runs in time $\tilde O(n d k^2)$.
\end{theorem}

\begin{algorithm}
\caption{DynamicProgram-kMedian($T, w, c$)}\label{alg:dynKMed}
\begin{algorithmic}[1]
\State Input: A quadtree $T$, a weight function on $T$'s node $w$, and a cell $c$.
\State Ouput: For each $k' \leq k$, a value $v_{k'}$ and a solution $S_{k'}$ for $k'$-median on $T(c)$
\State set $v_0 \gets w(c)\cdot \diam(c)$ and $S_0 \gets \emptyset$. 
\If{$w(c) < \frac{2 d \log n}{\eps}$} 
\State For all $k'$, set $v_{k'} \gets 0$ and $S_{k'} \gets k'$ copies of the center of $c$.
\Else 
\State Let $c_1, c_2$ be the two children of cell $c$
\State Let $v^i, S^i$ be the output of  DynamicProgram-kMedian($T, w, c_i$).
\State for all $k'$, let $(k_1, k_2) = \text{argmin}_{k_1 + k_2 = k'} v^1_{k_1} + v^2_{k_2}$, and do $v_{k'} \gets v^1_{k_1} + v^2_{k_2}$, $S_{k'} \gets S^1_{k_1} \cup S^2_{k_2}$.
\EndIf

\State \textbf{Return} $v, S$.
\end{algorithmic}
\end{algorithm}

\paragraph{Analysis}
\begin{lemma}\label{lem:dynKMed}
Step 3 of \cref{alg:kmed} computes a solution with expected  cost $\opt_T + k\cdot d^2 \log^2 n/ \eps \cdot \Lambda$, where $\opt_T$ is the optimal solution on the metric induced by the tree $T$, and the expectation is taken over the realization of the variables $\laplacenoise$. 
\end{lemma}
\begin{proof}
In an HST metric, we have the following property. For a cell $c$, three points $x, y \in T(c)$ and $z \notin T(c)$, it holds that $\dist(x, y) \leq \diam(T(c)) \leq \dist(x, z)$. Hence, if there is a center in $T(c)$, then all clients of $T(c)$ can be served by some center in $T(c)$. Moreover, if there is no center in $T(c)$, points in $T(c)$ are at distance of at least $\diam(T(c))$ from a center.

Let $F(c, k')$ be the expected cost of Dynamic\-Program-\-k\-Median($T, k', c$), 
where the probability is taken over the privacy randomness.
and $S$ be any solution. We show by induction that
for any cell $c$ of height $h$ in the tree with $T(c) \cap S \neq \emptyset$, 
{\setlength{\emergencystretch}{2.5em}\par}
$$F(c, |S \cap T(c)|) \leq \cost_T(T(c), S) + \frac{d \log n}{\eps} \cdot \Lambda \cdot |S \cap T(c)| \cdot h.$$
This is true by design of DynamicProgram-kMedian for any leaf that contains a center of $S$.

For an internal node $c$ with two children $c_1, c_2$ such that $T(c_1) \cap S \neq \emptyset$ and $T(c_2) \cap S \neq \emptyset$:
it holds that $F(c, S \cap T(c)) \leq F(c_1, |S \cap T(c_1)|) + F(c_2, |S \cap T(c_2)|)$. Hence, by induction:
\begin{align*}
    &F(c, |S \cap T(c)|) \leq F(c_1, |S \cap T(c_1)|) + F(c_2, |S \cap T(c_2)|)\\
    &\leq \cost_T(T(c_1), S) + \frac{d \log n}{\eps} \cdot \Lambda \cdot |S \cap T(c_1)| \cdot (h-1)\\ 
    &\qquad + \cost_T(T(c_2), S) + \frac{d \log n}{\eps} \cdot \Lambda \cdot |S \cap T(c_2)| \cdot (h-1)\\
    &\leq \cost_T(T(c), S) + \frac{d \log n}{\eps} \cdot \Lambda \cdot |S \cap T(c)| \cdot h,
\end{align*}
where the last lines use $\cost(T(c), S) = \cost(T(c_1), S) + \cost(T(c_2), S)$.
{\setlength{\emergencystretch}{2.5em}\par}

The last case is when $S \cap T(c_1) = \emptyset, S \cap T(c_2) \neq \emptyset$. Let $h$ be the height of $c$. We have $S \cap T(c_2) = S \cap T(c)$, and so: 
\begin{align*}
    F(&c, |S \cap T(c)|) \leq F(c_1, 0) + F(c_2,  |T(c_2) \cap S|)\\
    &\leq T(c_1) \cdot \diam(c_1) + \E[\laplacenoise(\eps / (d \log n)) \cdot \diam(c_1)] \\
    &\qquad + \cost_T(T(c_2), S) + \frac{d \log n}{\eps} \cdot \Lambda \cdot |S \cap T(c_2)| \cdot (h-1)\\
    &\leq \cost_T(T(c), S) + \frac{d \log n}{\eps} \cdot |S \cap T(c)| \cdot \Lambda \cdot h.
\end{align*}

This shows that the value computed by the algorithm is at most $\cost_T(\opt) + \frac{d^2 \log^2 n}{\eps} \cdot \Lambda \cdot k$. Now, we need to show the converse: the value computed corresponds to an actual solution. 

This is done inductively as well. For any $k'$ and cell $c$ one can compute a solution $S$ for $T(c)$ with $k'$ centers and expected cost at most $F(c, k') + \frac{d \log n}{\eps} \cdot k'$. For that, the base cases are when $k' = 0$, and then $\emptyset$ works, or when $w(c) \leq \frac{2d\log n}{\eps}$, where the center of the cell works. Otherwise,
it is enough to find $k_1, k_2$ such that $k_1+k_2 = k'$ and $F(c, k') = F(c_1, k_1) + F(c_2, k_2)$. Let $S_i$ be the solution computed for $T(c_i)$ with $k_i$ centers: the solution for $T(c)$ is simply $S_1 \cup S_2$. By induction, its cost is at most 
\begin{align*}
    F(c_1, k_1) + F(c_2, k_2) + \frac{d \log n}{\eps} \cdot \Lambda \cdot (k_1 + k_2) \\
    = F(c, k') +  \frac{d \log n}{\eps}\cdot \Lambda  \cdot k'. \qquad \qedhere
\end{align*}
\end{proof}

\begin{restatable}{lemma}{makeprivateDP}
\label{lem:makePrivateDP}
Algorithm~\ref{alg:makePrivate} is $\eps$-DP.
\end{restatable}
\begin{proof}
We show by induction that for a tree $T$ of depth $h$, MakePrivate($T, P$) is $\left(\frac{\eps}{d \log n} \cdot h\right)$-DP.

When the root of the tree has diameter at most $\Lambda/n$, the algorithm returns the zero function, which is $0$-DP.
Let $\calD$ be a tree of depth $h$ rooted at $r$ with $\diam(r) > \Lambda/n)$, and let $r_1, r_2$ be the two children of $r$. Computing $w(r)$ is $\frac{\eps}{d \log n}$-DP, by property of the Laplace Mechanism.

Now, by induction hypothesis, MakePrivate($T(r_1), P$) and
MakePrivate($T(r_2), P$) are $\left(\frac{\eps}{d \log n} \cdot (h-1)\right)$-DP. Since they are computed on two disjoint sets, the union of the two results is  $\left(\frac{\eps}{d \log n} \cdot (h-1)\right)$-DP as well. 
Notice that the algorithm MakePrivate($T, P$) boils down to computing $w(r)$,  MakePrivate($T(r_1), P$) and  MakePrivate($T(r_2), P$).
Hence, by composition MakePrivate($T, P$) is $\left(\frac{\eps}{d \log n} \cdot h\right)$-DP.{\setlength{\emergencystretch}{2.5em}\par}
\end{proof}

Combining Lemmas~\ref{lem:dynKMed}, \ref{lem:makePrivateDP} and properties of quadtrees, we conclude the proof of \cref{thm:kmed}:
\begin{proof}[Proof of \cref{thm:kmed}]
We start by proving the approximation guarantee. For this, note that the key property of quadtrees is that  
$\E_\calD[\dist_\calD(p, q)] \leq O\left(d^{3/2} \log n\right) \dist(p, q)$, where the expectation is taken on the tree randomness.
Hence, the optimal $k$-median solution is only distorted by a $O\left(d^{3/2} \log n\right)$ factor: $\opt_T \leq O\left(d^{3/2} \log n\right) \opt$. 

Combined with Lemma \ref{lem:dynKMed}, this shows the approximation guarantee of the whole algorithm.
Lemma \ref{lem:makePrivateDP} shows the privacy guarantee. 
What therefore remains is to bound the running time. 

Computing the cells of the quadtree containing some points of $P$ can be done in a top-down manner in time $O(nd \log n)$ as follows. 
Let $c$ be a cell at depth $d\cdot i + j$ with  $j < d$, and $c_1, c_2$ be the two children of $c$. Given a $T(c) \cap P$, it is easy to compute $T(c_1) \cap P$ and $T(c_2) \cap P$ in time $O(|T(c) \cap P|)$, by partitioning $T(c) \cap P$ according to the value of their $j$-th coordinate. Since there are $O(d \log n)$ levels in the tree, this is done in time $O(nd \log n)$.

Hence, the running time of \cref{alg:makePrivate} is bounded by $\tilde O(nd)$ plus the time to process empty cells added to $Q$.
There are at most $nd \log n$ empty cells with a non-empty parent added -- one per level of the tree and per point of $P$. Each of them gives rise to a Galton-Watson process: each node adds its two children with probability $\Pr[\laplacenoise(d \log n / \eps) > 2d \log n / \eps] = e^{-2} < 1/2$. By standard properties of a Galton-Watson process, this goes on for a constant number of steps. 
Therefore, there are at most $\tilde O(n d)$ empty cells added to $Q$, which concludes the running time bound for \cref{alg:makePrivate}. 

Let $N$ be the number of cells that have a non-zero value of $w$. We claim that $N = \tilde O(nd)$ and that the running time of \cref{alg:dynKMed} is $O(Nk^2)$. 
For the first claim, note that $N$ is equal to the number of cells added to $Q$, which is $\tilde O(nd)$ as explained previously. For the second claim, notice that there are at most $kN$ different calls to DynamicProgram-kMedian, each being treated in time $O(k)$. Hence, the complexity of \cref{alg:dynKMed} is $O(Nk^2) = \tilde O(ndk^2)$. This concludes the proof.
\end{proof}

\section{MPC Implementation}
\paragraph{Brief description of MPC}\label{sec:mpc}

We briefly summarize the MPC model \cite{beame2017communication}. The input data has size $N = nd$, where $n$ is the number of points, and $d$ the dimension. We have $m$ machines, each with local memory $s$ in terms of words (of $O(\log(ms))$ bits). We assume that each word can store one input point dimension. We work in the {\it fully-scalable} MPC framework~\cite{andoni2018parallel} where the memory is sublinear in $N$. 
More precisely, the machine memory is $s=\Omega\left(N^{\delta}\right)$ for some constant $\delta\in (0,1)$, and the number of machines $m$ is such that $m\cdot s=\Omega\left(N^{1+\gamma}\right)$, for some $\gamma > 0$.

The communication between machines is as follows.
At the beginning of the computation, the input data is distributed arbitrarily in the local memory of machines, the computation proceeds in parallel rounds where each machine can send (and receive) arbitrary messages to any machine, subject to the total messages space of the messages received (or sent) is less than $s$. In case some machine receives more than $s$ messages, the whole algorithm fails.

For our MPC algorithm we assume that $k\ll s$. This ensures that the final solution of size $kd$ fits in the memory of one machine, which is common for real world applications.  

More formally we assume that there are $m =  \Omega(n^{1-\delta+\gamma})$ machines each with memory $s = \Omega(n^{\delta} d\log(n))$, and $k \leq n^{\gamma}$, with $\delta - \gamma > \eps$ for some constant $\eps$. 

In that section, we first show the following \emph{low dimensional} theorem:
\begin{theorem}
\mbox{}\label{thm:kmedMPC}
Assuming  $k \leq n^{\gamma}$, there exists a $O(d \log n)$ rounds MPC algorithm using $m =  O(n^{1-\delta+\gamma})$ machines each with memory $s = O(n^{\delta} d\log(n))$ that simulates exactly the private $k$-median from Theorems~\ref{thm:kmed}.
\end{theorem}

This algorithm is suited for low dimensional spaces, as the number of rounds depends on $d$. We show in Section~\ref{sec:kmedMPC} how to replace this dependency by a $O(\log k)$, both in the number of rounds and in the approximation ratio.

We then show how to use dimension reduction, to replace dependencies in $d$ by $\log k$:
\begin{theorem}
\mbox{}\label{thm:kmedMPCHD}
Assuming  $k \leq n^{\gamma}$, there exists a $O(\log k \cdot \log n)$ rounds MPC algorithm using $m =  O(n^{1-\delta+\gamma})$ machines each with memory $s = O(n^{\delta} d\log(n))$ that computes a solution to $k$-median with cost at most
\[O\left(\log^{3/2} k \log n \cdot \opt + \frac{\log^{3} k \log^{3} n \cdot k}{\eps} \cdot \Lambda + \frac{kd \log k}{\eps}\cdot \Lambda \right).\]
\end{theorem}

\subsection{Algorithm for Low Dimensional Inputs}
We now describe a high level view of our algorithm which as we can prove simulates exactly (with high probability) our private $k$-median algorithm. The algorithm uses a shared hash function $h$ to compute the quadtree consistently over the machines. 
Informally, first, each machine computes over the points stored, all the cells which the points belong to in the tree at each level. 
To compute the total count of each cell, one can use the algorithm from Andoni et al.~\cite{andoni2018parallel} (section E.3 of the arxiv version), that computes in a constant number of rounds the number of points in each cell. At the end of that algorithm, the size of each cell is stored in some unspecified machine. To organize the quadtree data in order to be able to process it, we use a shared function $r$ such that a machine $r(c)$ is responsible for all computations related to cell $c$. We will need care to ensure that no machine is responsible for more cell than what its memory allows.

Then the computation proceeds bottom-up solving the dynamic programming problem in $O(d \log n)$ rounds.\footnote{We note that a more careful and intricate implementation of the dynamic program that requires only $O(d)$ rounds can be achieved. We decided to chose simplicity rather than saving one log factor.} Finally, the computation proceeds over the tree top-down in other  $O(d \log n)$ rounds to extract the solution.

\begin{algorithm}
\caption{MPC-quadtree($P$)}\label{alg:mpc-aggregation}
\begin{algorithmic}[1]
\State Each machine receives an arbitrary set of $n^\delta$ points of $P$.
\State Each machine, for each point $p$ received, computes, using $h_s$, the quadtree cell $c_i(p)$ in which the point $p$ is at level $i \in [d\log(n)]$. 
\State Compute the count of every cell.
\State Send the count of cell $c$ to machine $r(c)$, for all $c$.
\end{algorithmic}
\end{algorithm}

Using the algorithm from Andoni et al.~\cite{andoni2018parallel}, one can compute in $O(1)$ steps the count for each cell of the quadtree. Hence, we have the following result:

\begin{fact}
\cref{alg:mpc-aggregation} runs in $O(1)$ many rounds.
\end{fact}

At the end of \cref{alg:mpc-aggregation}, we are given a quadtree, represented as follows: each cell $c$ is represented by a machine $r(c)$, which stores a count of input nodes in the cell and pointers towards each children. $r$ is a surjection from a set of $O(nd \log n)$ cells to $m$ machine: we chose it in order to ensure that for any machine $\mathcal{M}$, $|r^{-1}(\mathcal{M})| \leq \frac{O(nd \log n)}{m}$.

We now explain in more details how to implement the algorithm from \cref{thm:kmed}, given that representation of the quadtree. 

First, it is straightforward to implement \cref{alg:makePrivate} in $1$ rounds -- as each cell only needs to compute the DP count of points in the cell. Next, \cref{alg:dynKMed} is straightforwardly implemented  in $O(d \log n)$ rounds, as computing the output vector $v$ of the dynamic program for a cell only requires knowing those of its children -- and it is therefore easy to simulate bottom-up the dynamic program.

What remains to be proven is that no machine gets responsible for more cell than it can afford in memory. More precisely, every time a machine is responsible for a cell, it stores $O(k)$ memory words, for the simulation of the dynamic program.
Hence, we need to show that no machine is responsible for more than $s/k$ many cells.

\begin{fact}
No machine is responsible for more than $s/k$ many cell.
\end{fact}
\begin{proof}
Our choice of $m$ and mapping $r$ ensures that a given machine gets responsible for at most $\frac{O(nd \log n)}{m} = O\left(n^{\delta - \gamma} d \log n\right)$. Similarly, our constraints on $k$ and $s$ ensures $\frac{s}{k} = \Omega(n^{\delta - \gamma} d \log n)$, which concludes the proof.
\end{proof}

Combining those two facts concludes \cref{thm:kmedMPC}.

\newcommand{\coreset}{\Omega}
\newcommand{\subApprox}{\mathfrak{a}}
\newcommand{\calC}{\mathcal{C}}
\newcommand{\calS}{\mathcal{S}}
\newcommand{\polylog}{\text{polylog}}

\subsection{$k$-Median in $O(\log n)$-MPC rounds via dimension reduction}\label{sec:kmedMPC}
The goal of this section is to use standard dimension-reduction techniques to remove the dependency in the dimension from \cref{thm:kmedMPC} and show \cref{thm:kmedMPCHD}.

For that, one can use dimension reduction techniques to project the dataset onto $O(\log k)$ dimensions, while preserving the cost of any clustering.

However, the output of our algorithm should be a set of centers in $\R^d$, and not a clustering:
an additional step is therefore needed, once clusters have been computed in $\R^{O(\log k)}$, to \textit{project back} and find centers in the original space. For $k$-means, this can easily be done using differentially-private mean \cite{badih_approximation}. We show how to perform the equivalent for $k$-median.

We draw here a connection with the coreset literature. More precisely, we leverage results from Cohen-Addad et al.~\cite{Cohen-AddadSS21}, who showed how to compute an approximate solution to $1$-median by only considering an uniform sample of constant size. Therefore, in the MPC setting it is enough to sample a constant number of points from each cluster computed in low dimension, and send them to a machine that can compute a median for them in the original high dimensional space.

For that last step, we rely on the following result.
\begin{lemma}[Corollary 54 in \cite{badih_approximation}]
\mbox{}\label{lem:DP1median}
For every $\eps >  0$, there is an $\eps$-DP polynomial time algorithm for $1$-median such that, with probability $1-\beta$, the additive error is $O\left(\frac{d \Lambda}{\eps} \text{polylog}\left(\frac{1}{\beta}\right)\right)$
\end{lemma}

We consider the following algorithm, a simplified variant of Algorithm 1 in Cohen-Addad et al.~\cite{Cohen-AddadSS21}.

\begin{algorithm}
\begin{algorithmic}
\State \textbf{Input:} A dataset $P$, an $\alpha$-approximate median $\subApprox$ for $P$ with cost $\calC$, and parameters $t, d_{close}, r_{small}$.
\State 1. Sample a set $\coreset$ of $t$ points uniformly at random.
\State 2. Remove from $\coreset$ all points at distance less than $\Delta = \frac{d_{close}}{\alpha}\cdot \frac{\calC}{|P|}$, and add to $\coreset$ the point $\subApprox$ with multiplicity equal to the number of removed points.
\State 3. Define rings $R_i$ such that $R_i\cap \coreset$ contains all the points at distance $(2^{i} \cdot \Delta, 2^{i+1}\cdot \Delta]$ from $\subApprox$, for $i \in \{1,..., \log(|P| \alpha / \mu_2)\}$.  Let $R_0$ be $\{\subApprox\}$, with multiplicity defined in step 2.
\State 4. If $|R_i\cap \coreset|< r_{small}\cdot |\coreset| + \laplacenoise(1 / \eps)$, remove all points in $R_i\cap \coreset$ from $\coreset$.
\State 5. Solve the problem on the  set $\coreset$, using the algorithm given by \cref{lem:DP1median} with $\beta = 1/k$.
\end{algorithmic}
 \caption{Finding the median via uniform sampling}
 \label{alg:coreset}
\end{algorithm}

\begin{lemma}
Algorithm \cref{alg:coreset} is $2\eps$-DP.
\end{lemma}
\begin{proof}
First, the set of rings selected at step 4 is $\eps$-DP: the selection of one ring is $\eps$-DP, by Laplace mechanism, and since the rings are disjoint the composition of DP mechanisms ensures that the full set of selected rings is $\eps$-DP.

Now, given a selected set of rings, the set $\coreset$ varies by at most one point when the input $P$ varies by a single point.
Since the algorithm used in step 5 is $\eps$-DP, by composition, the whole algorithm is $2\eps$-DP.
\end{proof}

As shown by Cohen-Addad et al.~\cite{Cohen-AddadSS21}, this algorithm computes an $O(1)$-approximation to $1$-median on $P$, with $t = \polylog(|P|)$. 
Hence, we can easily use it to \emph{project back} the centers, and conclude the proof of \cref{thm:kmedMPCHD}.

\begin{proof}[Proof of \cref{thm:kmedMPCHD}.]
Using Johnshon-Lindenstrauss lemma, it is possible to project the points onto a space of dimension $\tilde d = O(\log k)$, preserving the cost of any clustering up to a constant factor (see Makarychev et al.~\cite{makarychev2019performance}). 
In that projected space, the algorithm from \cref{thm:kmedMPC} computes privately a solution with cost $O(\log^{3/2} k \log n) \cdot \opt + \frac{\log^{3} k \log^{3} n \cdot k}{\eps} \cdot \Lambda$, but centers are not points in $\R^d$ -- they are nodes of the quadtree.

To compute good centers in $\R^d$ from the quadtree solution, we use \cref{alg:coreset}: in each cluster induced by the quadtree solution, sample the set $\coreset$. Since $\coreset$ has size $O(\log^{3} k \log^2 n)$, it can be sent it to a centralizing machine, that in turn can run \cref{alg:coreset}. The additional additive error is $O\left(\frac{d \Lambda \log k}{\eps}\right)$ in any cluster, hence in total $O\left(\frac{kd\Lambda \log k}{\eps}\right)$. To sample the set $\coreset$, each machine can send to the centralizing one the number of points it stores from $P$, and the centralizing computes the number of points to be sampled in each machine.\footnote{For instance, the centralizing machine can sample as set $R$ of $\mu_1$ points from $\{1,...,|P|\}$. Then, if machine $i$ stores $n_i$ points from $P$, it computes a uniform sample $R \cap (\sum_{j < i} n_i, \sum_{j \leq i} n_i]$ many points. The union of those sample is uniform.}
\end{proof}

\section{Extension to $k$-Means}\label{sec:kmeans}
The main focus in the paper is on k-median, however we can also show an extension of our result for $k$-means: 
\begin{restatable}{theorem}{km}
\label{thm:km}
There exists an $\eps$-DP algorithm $A$ that takes as input
  a set of points and computes a solution for $k$-means with at most $(1+\alpha)k$ centers and, with probability $3/4$, costs at most $\poly(d, \log n, 1/\alpha) \cdot \opt + k d^2 \log^2 n / \eps \cdot \Lambda^2$.
\end{restatable}

We give an in-depth description with full proofs in \cref{ap:kmeans}. Here, we outline the high-level ideas, where we show as well how to remove the extra $\alpha k$ centers, to get an approximate solution with exactly $k$ centers. 
As explained in the introduction, we establish the following lemma, that shows how we can improve a solution given as input.

\begin{restatable}{lemma}{mainkmeans}
\label{lem:main:kmeans}
  Given an arbitrary solution $L$, there exist an $\eps$-DP algorithm $A$ that takes as input a set of points and computes a solution for $k$-means with at most $k + \frac{\alpha}{2}\cdot|L|$ centers and, with probability $1-\pi$, costs at most $\frac{O(d^9 \log^2 n)}{\alpha^6 \pi^6} \cdot \opt + \alpha \cdot \cost(L) + k d^2 \log^2 n / \eps \cdot \Lambda^2$.
\end{restatable}

Although the quadtree decomposition approximates distances well in expectation, it works poorly for squared distances. Indeed, two points $p, q$ have probability $\frac{d \cdot \dist(p, q)}{2^i}$ to be cut at level $i$: hence, the expected distance squared between $p$ and $q$ is 
$d \cdot \dist(p, q) \cdot \sum_i \sqrt d 2^i$, which means that the distance squared can be distorted by an arbitrarily large factor in expectation.

 However, observe that $p$ and $q$ have tiny probability to be cut at a level way higher than $\log (d\cdot \dist(p, q))$. Hence, there is a tiny probability that points are cut from their optimal center at a high-level. The question is then: what to do when this happens? Here we want to avoid routing in the tree since the squared distance could be arbitrarily large and we may want to deal with such points in a different way.
 To do so, we use a baseline solution $L$ to guide our decisions on points for which the tree distance to their closest center in the optimum solution badly approximates the true distance, let call them \emph{bad points}.
Since we don't know the optimum solution, we don't know the bad points and so we will use $L$ as a proxy for finding the potential bad points.

We show that the solution computed by our algorithm is good w.r.t. to a solution that contains all facilities of $L$  for which the quadtree distances are not a good approximation of the true distances. We call those facilities \emph{badly-cut}. To bound the cost of a client $c$, we distinguish three cases. Either the distance from a point to the optimal center is good in the tree, and we are happy because we can serve it nicely in the tree. Or its closest center of $L$ is not badly-cut, in which case we argue that the distance to the optimal center cannot be too high compared to its optimal cost. In the last case, where the closest center of $L$ is badly-cut, we simply assign the point to $L$ since we are working with a solution containing all centers of $L$. This happens with some tiny probability, and will not be too costly overall, i.e.: only a tiny fraction of the cost of $L$.
\

\section{Empirical Evaluation}
\label{sec:exp}

\begin{figure}[t]
\centering 
\begin{tabular}{cc}
\includegraphics[width=.44\linewidth]{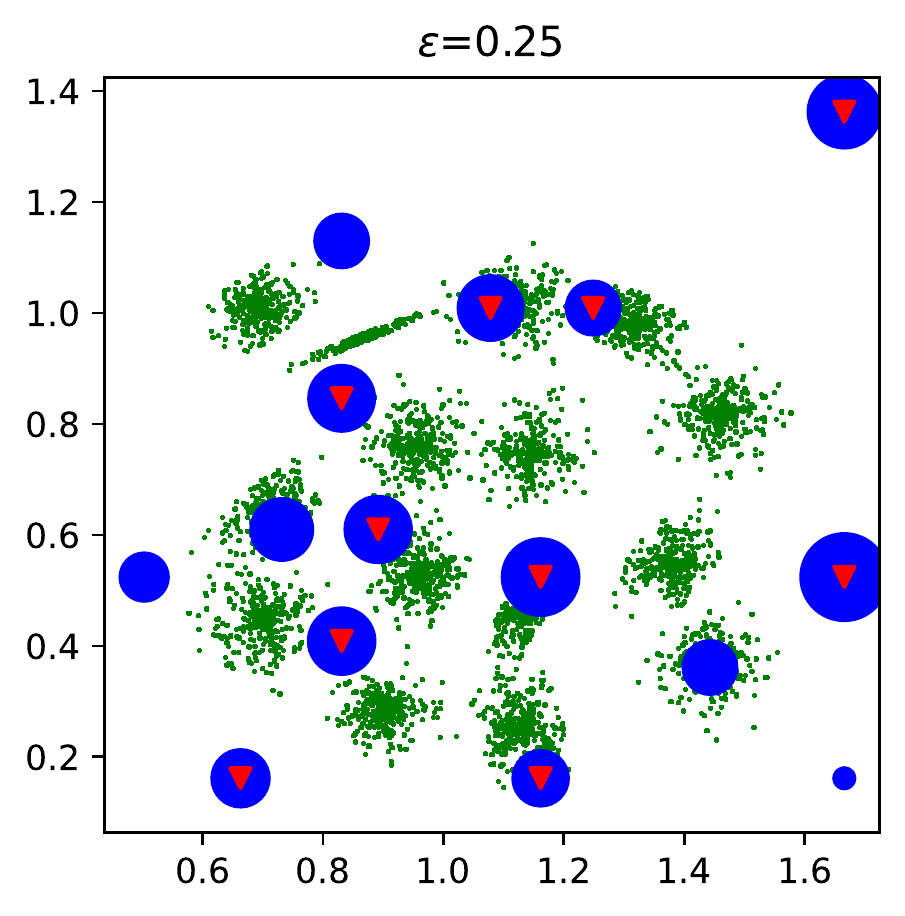}
&\includegraphics[width=.44 \linewidth]{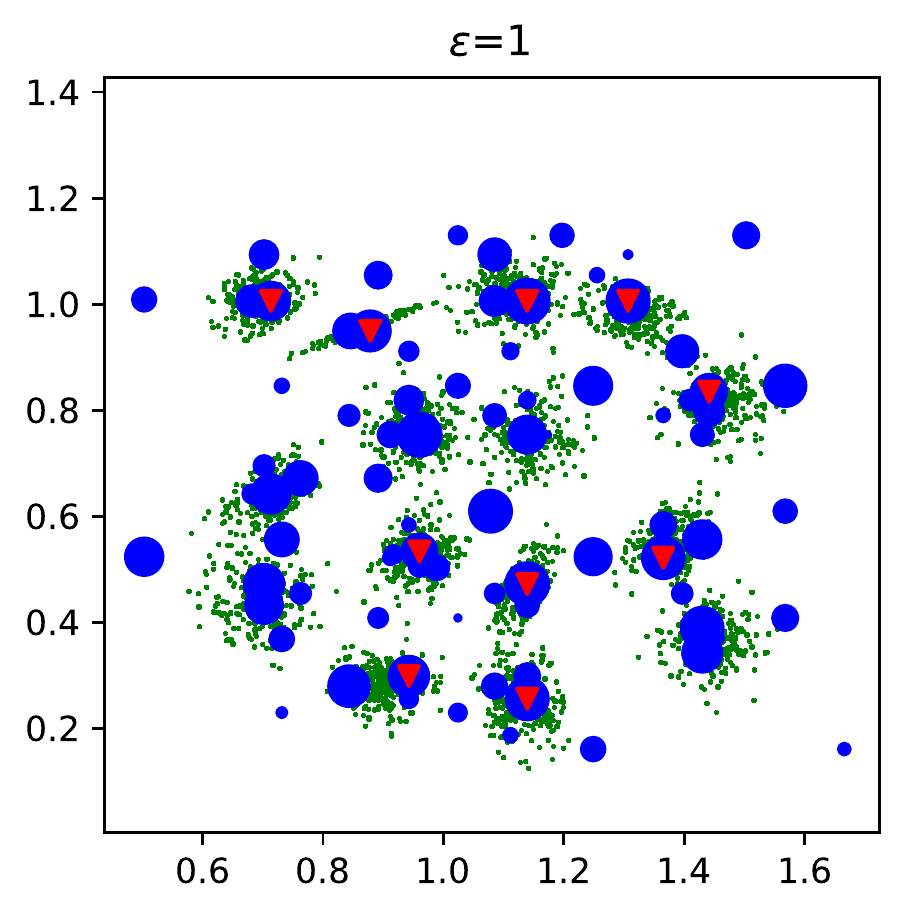}\\
\vspace{-0.1in}
(a) & (b)  
\end{tabular}
\caption{Visualization of our algorithm. Original dataset in green. Leaves of the tree scaled by their weight in blue and centers found by our algorithm in red. (a) $\epsilon = 0.25$ and (b) $\epsilon = 1.0$.}
\label{fig:visualization}
\end{figure}

\begin{figure}[t!]
\centering 
\begin{tabular}{cc}
\includegraphics[width=.44 \linewidth]{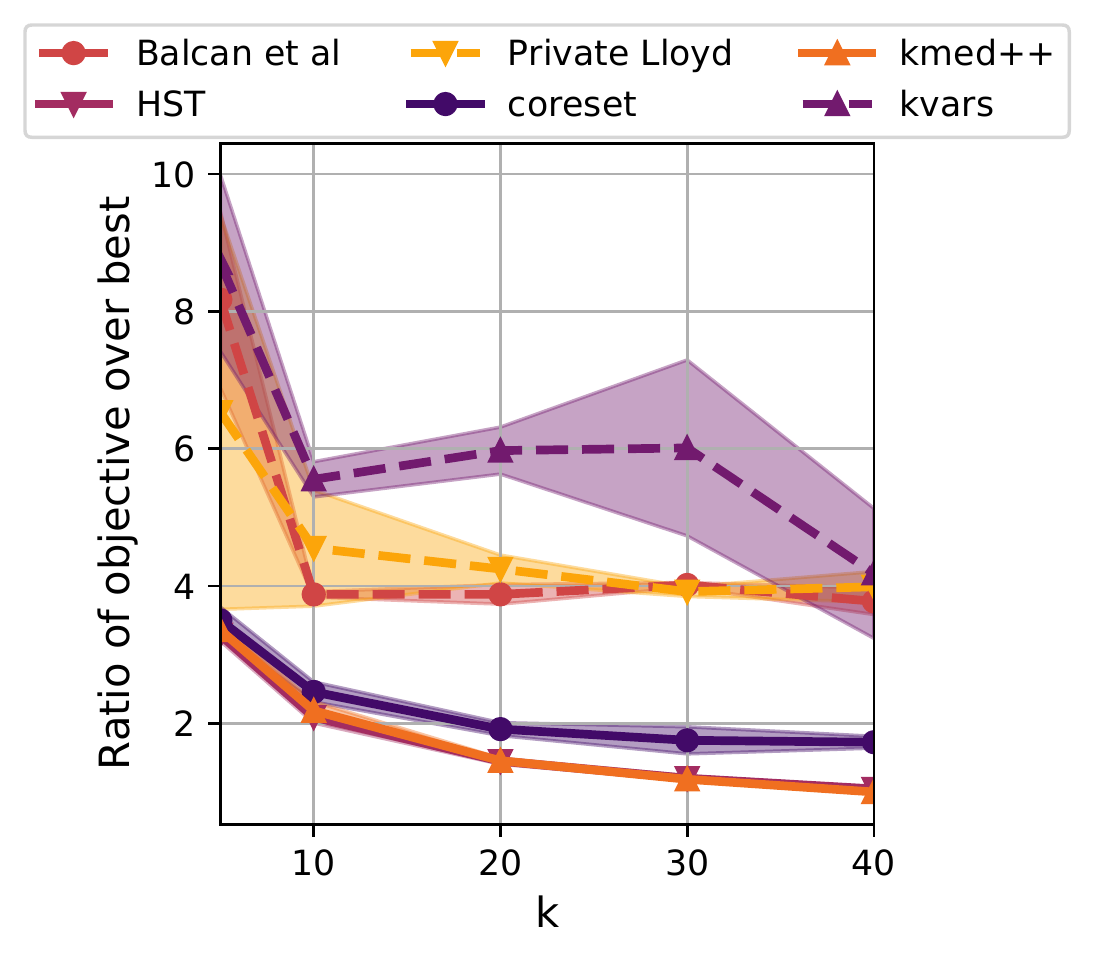} 
&\includegraphics[width=.44 \linewidth]{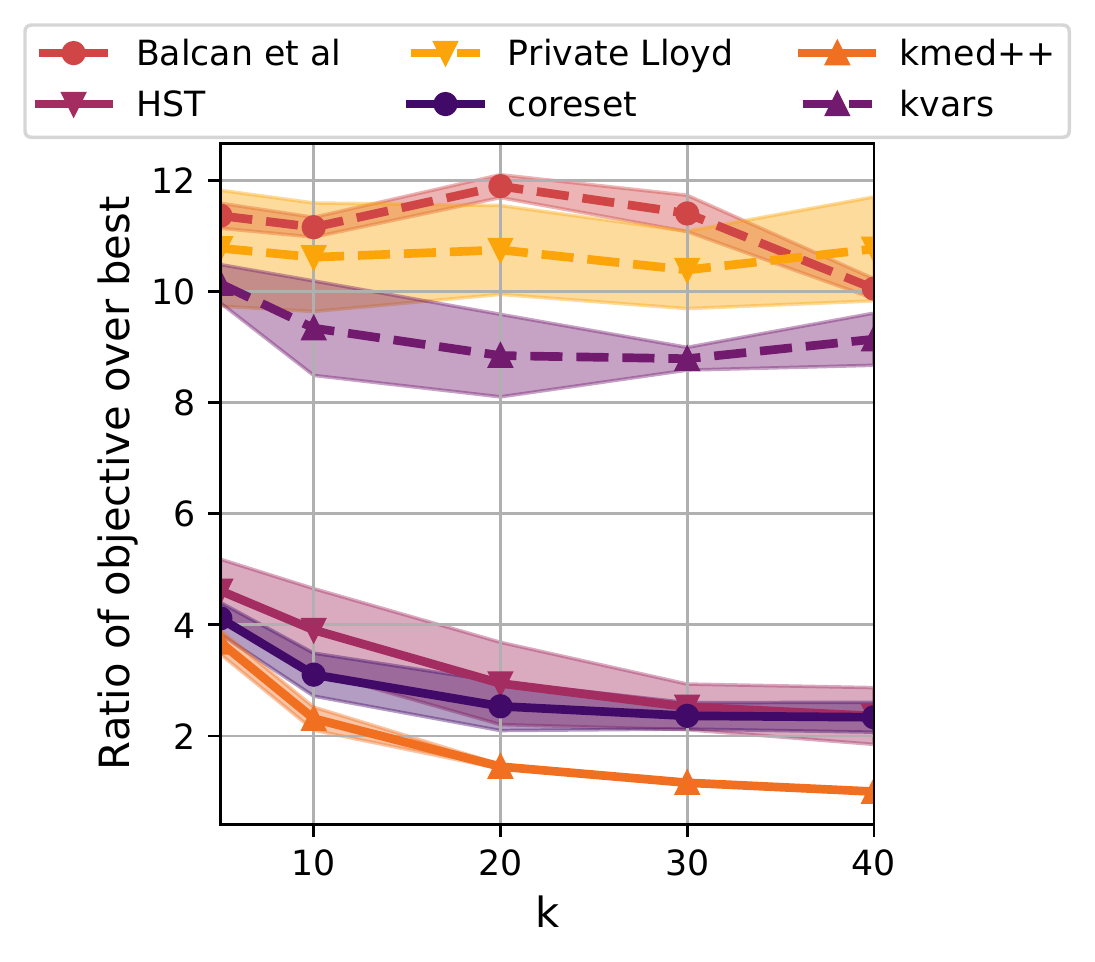}\\
\vspace{-0.1in}
(a) & (b) 
\end{tabular}
\caption{Comparison of algorithms on (a) \skin and (b) \shuttle}
\label{fig:smallscale}
\end{figure}

In this section, we present an empirical evaluation of our algorithm for the $k$-median objective. To the best of our knowledge, this is the first comprehensive empirical evaluation of private k-median algorithms as the majority of experimental results has previously focused on $k$-means. All datasets used here are  publicly-available, and the code accompanying our paper can be found at this page: \url{https://github.com/google-research/google-research/tree/master/hst_clustering}

\textbf{Datasets.} 
We used the following well known, real-world datasets from the UCI Repository~\cite{Dua:2019} that are standard in clustering experiments \skin~\cite{skindataset} ($n=245057, d=4$),  \shuttle~\cite{Dua:2019} ($n=58000, d=9$), \covtype~\cite{blackard1999comparative} ($n=581012, d=54$)
and \higgs~\cite{higgs} ($n=11000000$, $d=28$). Finally, we use a  publicly available synthetic datasets \synth ($n=5000$, $d=2$)~\cite{ClusteringDatasets} for visualizing clustering results.

\textbf{Experimental details.}
To simplify the stopping condition of Algorithm~\ref{alg:makePrivate} we parameterize our algorithm by a depth parameter $\alpha$ and weight parameter $\beta$. We grow all of our trees to a max depth of $\alpha d$ and stop splitting the tree when $w(c) < \frac{10 \beta d}{\epsilon}$ instead of $2 d \log n/\epsilon$. This threshold was chosen to decrease the chance of potentially splitting empty cells multiple times and does not affect the privacy properties of the mechanism. The implementation for building the tree embedding was done using C++ in a large-scale distributed infrastructure. The dynamic program for solving the optimization problem in the tree was done in a single machine. 

\paragraph{Non-private baseline}
We compare the results of our algorithm against a non-private implementation of $k$-median++~\cite{arthur2007k} (\textbf{kmed++}) with 10 iterations of Lloyd's algorithm. Each iteration was done by optimizing the $k$-median objective exactly using Python's BFGS optimizer. 

\paragraph{Private baselines}
To the best of our knowledge all private baselines for clustering algorithms have focused on the k-means problem. However, using a private $1$-median algorithm it is possible to adapt some of the prior work to solve the private k-median problem.

As a first step we implement a subroutine of the $1$-median problem using the objective perturbation framework of \cite{TPDPCO}. The algorithm described in \cite{TPDPCO} requires a smooth loss function. We therefore modified the k-median objective to the  $\frac{1}{\lambda}$-smooth k-median objective $f_\lambda \colon x \mapsto \|x\| + 2 \lambda \log \big((1 + e^{-\frac{\|x\|}{\lambda}})/2\big)$ which converges to $\|x\|$ as $\lambda \mapsto 0$. 
Given this tool, we implemented the following algorithms.

    $\bullet$ {\bf HST }: The MPC version of Algorithm~\ref{alg:kmed}. After finding the centers using the tree, we ran 4 iterations of the Lloyd algorithm using the private $1$-median implementation described above using at most 20k points from each cluster for the optimization step to allow it to fit in memory. We split the privacy budget $\epsilon$ uniformly: using $\epsilon/5$ to build the tree and $\epsilon/5$ per Lloyd's iteration. We tune the parameters $\alpha \in \{10, 12, 14\}$ and $\beta \in \{6, 8, 10\}$. The hyper-parameters for the $1$-median solver were set to $\lambda = 0.2$ and $\gamma$ to $0.01*\sqrt{d}/n$ ($\gamma$ is a bound on the gradient norm of the optimizer defined in \cite{TPDPCO}).

$\bullet$ \textbf{Private Lloyd}: a private implementation of Lloyd's algorithm. This algorithm has no approximation guarantee. The initial centers are chosen randomly in the space, and at each iteration, each point is assigned to the nearest center, and centers are recomputed using the private 1-median algorithm. We chose the number of iteration to be 7, as a tradeoff between the quality of approximation found and the privacy noise added. Here, the hyper-parameters for the $1$-median solver were $\lambda = 1$ and $\gamma = 0.01 \sqrt{d}/n$.

$\bullet$ \textbf{Balcan et al}: the private algorithm of~\cite{balcan}. The solution computed has a worst case cost of at most $\log(n)^{3/2} \cdot \opt + \poly(d, k, \log n)$. We modified the code available online~\cite{balcanCode} to adapt it to $k$-median, by using our $1$-median implementation with $\lambda = 1$ and $\gamma = 0.01 \sqrt{d}/n$.

$\bullet$ \textbf{kvars}: A private instantiation of the kvariates heuristic algorithm of~\cite{kvariates}. The algorithm uses a sub-routine that splits data into computation nodes. We hash each point using  SimHash \cite{simhash} to assign them to one of 500 computation nodes.
 
$\bullet$  \textbf{Coreset}\footnote{\url{https://ai.googleblog.com/2021/10/practical-differentially-private.html}}: A  heuristic algorithm for private k-means clustering that creates a coreset via  recursive partitioning using locality sensitive hashing. We modified the heuristic to handle k-median with our private 1-median implementation, with $\lambda = 0.2$ and $\gamma=0.01\sqrt{d}/n$.

\paragraph{Other baselines not evaluated}
We describe here other potential candidate baselines which  we found not feasible to compare against. 
Since our work focuses on scalability, we do not compare against algorithms with impractically large running times like the algorithm of ~\cite{StemmerK18, badih_approximation} which have state-of-the-art theoretical approximations but that have not previously been implemented.\footnote{Private communication with the authors of \cite{badih_approximation} confirmed that there is no practical implementation of this algorithm available.}
We also did not compare against \cite{anamayclustering}, as it lacks guarantees for $k$-median and the baseline \cite{balcan} showed comparable performance with their algorithm. Finally, we do not compare with the heuristic GUPT \cite{gupt} as it does not provide an explicit aggregation procedure for k-median.

For all algorithms we report the average of $10$ runs. We varied the number of centers, $k$, from $5$ to $40$ and,  $\epsilon$,  from $0.25$ to $1$.

\paragraph{Results} We begin by showing a visualization of our algorithm on the \synth  dataset of 2 dimensions to give intuition on the effect of privacy on constructing the tree embedding. Figure~\ref{fig:visualization}(a) shows the centers returned by our algorithm for $\epsilon=0.25$ and $\epsilon=1$. It is immediate to see that as $\epsilon$ increases our tree embedding captures the geometry of the dataset correctly.

\begin{figure}[t] 
\centering
\begin{tabular}{cc}
\includegraphics[width=0.44\linewidth]{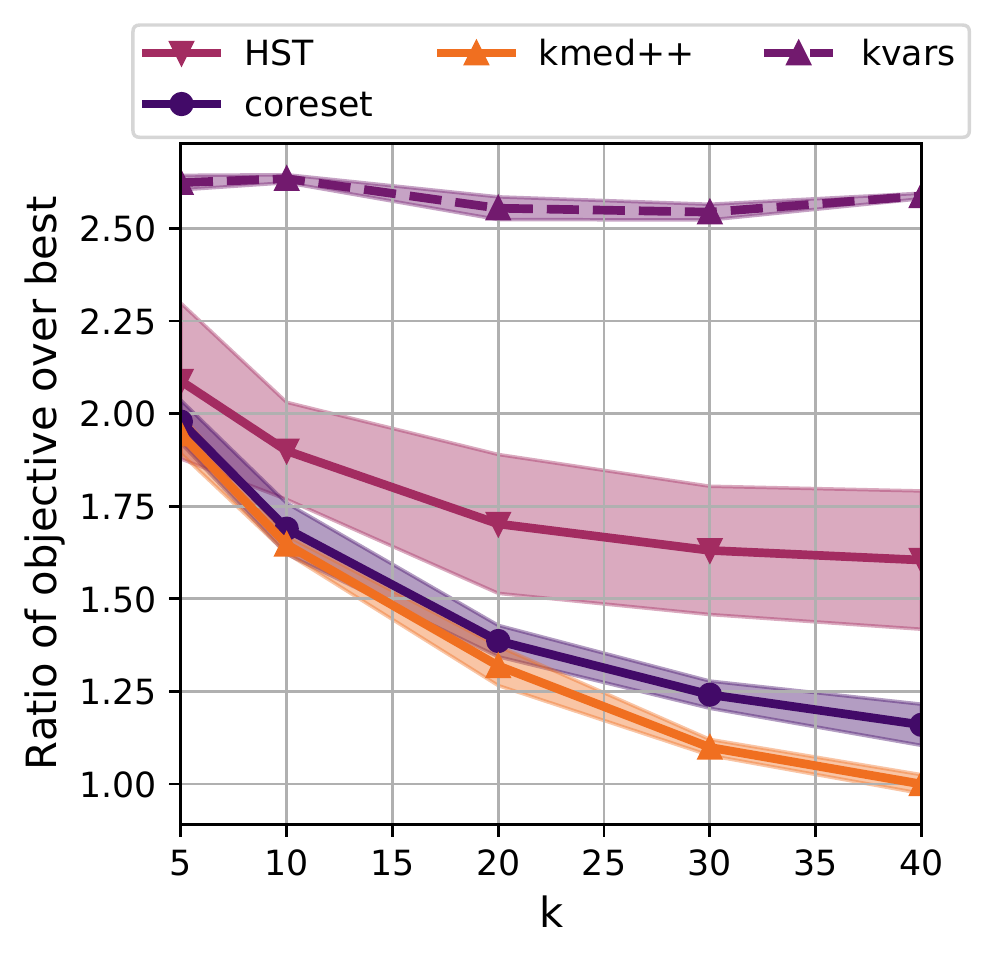} & 
\includegraphics[width=0.44\linewidth]{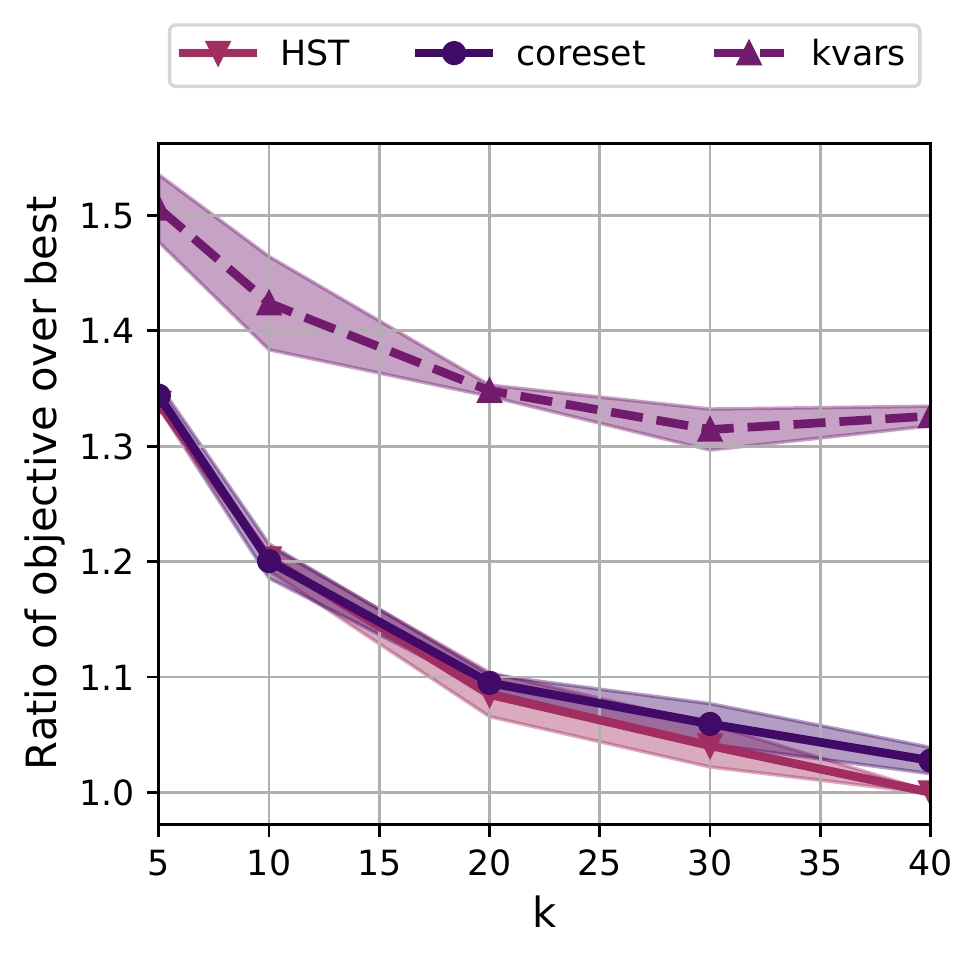}\\
\vspace{-0.1in}
(a) & (b)
\end{tabular}
\caption{Objective function as a function of $k$ datasets (a) \covtype and (b) \higgs.}
\label{fig:largescale}
\end{figure}

We now discuss the quality of the clusterings returned by each algorithm. We begin evaluating all baselines on the small datasets \skin and \shuttle. Figure~\ref{fig:smallscale} shows the quality of each algorithm for $\epsilon=0.5$. The plots are normalized by the best clustering objective. There are several points worth noting in this plot. First, the performance of the Balcan et al. algorithm which has the best approximation guarantees is consistently outperformed by our algorithm and the coreset algorithm. Second, notice that on \skin our algorithm achieves a performance that is essentially the same as the non-private baseline. 

For the large datasets \covtype and \higgs, it was impossible for us to run the Balcan et al.~approach. Therefore, we only compare our algorithm against the coreset and kvars baselines. Figure~\ref{fig:largescale} shows the results. Here we see that  our algorithm has the strongest performances on HIGGS while on \covtype it is comparable to the coreset heuristic and slightly worse for large $k$. 

We compare only the quality of the solutions computed and not the running time, as the parallel implementation has a large overhead and it never runs really fast. However, our implementation does run and provide apparently good results on large scale datasets on which other private algorithms do not terminate or give really poor results -- to the notable exception of the Coreset algorithm, which does not enjoy theoretical guarantees.

In summary, our empirical evaluation confirms that our approach, which is the only method that has both theoretical performance guarantees and can be made to scale to large datasets consistently performs well on a wide variety of examples, achieving accuracy much higher than the worst case analysis would indicate.

\section{Conclusion}
We present practical and scalable differentially private algorithms for $k$-median with worst case approximation guarantees. Although their worst-case performance is worse than state of the art methods, they are parallelizable, easy to implement in distributed settings, and empirically perform better than any other algorithm with approximation guarantees.
Furthermore, we present an extension of those algorithms to the $k$-means objective, with a theoretical analysis.
A natural open question is to  close this gap between theory and practice: finding scalable methods that have even better worst-case guarantees.

\bibliographystyle{plain}
\bibliography{references}

\newpage
\appendix

\section{Supplementary Material -- Extension to $k$-Means}\label{ap:kmeans}

In this section we prove \cref{thm:km}. 
\begin{proof}[Proof of \cref{thm:km}]
We consider the following algorithm: 

\begin{algorithm}
\caption{$k$-means algorithm with extra centers}
\begin{algorithmic}[1]
\State \textbf{Input:} A set of clients $X$.\textbf{Output:} A set of $k$-means centers $C$.
\State $L \gets {0}$, $\eps' \gets \eps/\log n$ 
\For{$\log n$ steps} 
 \State $L' \gets$ Solution computed by $A$ as described by Lemma~\ref{lem:main:kmeans}, setting $\pi = \frac{1}{4\log n}$, with privacy parameter $\eps'$.\footnote{Instead, one could apply $\log \log n$ times the algorithm of Lemma~\ref{lem:main:kmeans} with a constant probability $\pi$ and take the outcome of the best run. This slightly changes the parameters -- it saves a few $\log n$ in the approximation -- but for the proof we opted for simplicity rather than performances.} 
    \State $L \gets L'$
\EndFor
\State Return $L$
\end{algorithmic}
\end{algorithm}

We argue that the above algorithm produces a solution satisfying the claims
of Theorem~\ref{thm:km}. We have that the initial solution has cost at most $n \Lambda^2$.
Then, by repeatedly applying Lemma~\ref{lem:main:kmeans}, we obtain 
solutions of geometrically decreasing cost. 
More precisely, after $i$ iterations, we claim that with probability $1 - \frac{i}{4 \log n}$, $L$ has size at most $k \cdot \sum_{j = 0}^i \alpha^i$, and  the cost of $L$ is at most  $\poly(d, \log n, 1/\alpha)\cdot \opt + \alpha^i n \Lambda + kd^2 \log^2 n /\eps' \cdot \Lambda^2$ . This is true when $i = 0$, and follows directly from applying Lemma~\ref{lem:main:kmeans}.

It follows that the final solution computed after $\log n$ steps 
has cost at most $\poly(d, \log n)$ times $\opt$ plus additive
$kd^2 \log^3 n /\eps \cdot \Lambda^2$. 
Moreover, since $\eps' = \eps / \log n$, the algorithm is by composition $\eps$-DP.
Finally, the number of centers is at most 
$k(1+\alpha)$ as desired.
\end{proof}

Hence, the key is to prove \cref{lem:main:kmeans}. Before describing the ideas behind the extension to $k$-means, we introduce some notations.
We say that two points $p, q$ are \emph{cut at level $i$} when their lowest common ancestor
in the tree is at level in $(d\cdot (i-1), d \cdot i]$, i.e., the diameter of that common ancestor is in $(\sqrt d \cdot 2^{i-1}, \sqrt{d} \cdot 2^i]$. In that case, the 
distance in the quadtree metric between $p$ and $q$ is at most $\sqrt d 2^i$.
We say that a ball $B(p, r)$ is cut at level $i$ if $i$ is the largest integer such that
there exists a point $q$ with $\dist(p, q) \leq r$ and $p$ and $q$ are cut at level $i$.

Recall \cref{lem:ballCut}: 
  For any $i$, radius $r$ and point $p$, it holds that $\Pr[B(p, r) \text{ is cut at level } i] = O\left(\frac{d r}{2^i}\right)$.

\paragraph{Formalization} Let $P \subseteq \R^d$ be an instance of the 
$k$-means problem in $\R^d$.
Let $\globalS$ be an optimal solution to $P$ and $L$ be an arbitrary solution. 
For a given client $c$, we let
$L(c)$ (resp. $\globalS(c)$) denote the center of $L$ (resp. $\globalS$) 
that is the closest to $c$ in solutions
$L$ (resp. $\globalS$).

For a quadtree decomposition $\calD$, we say that a client $c$ is \emph{badly-cut} if the ball $B(c, \dist(c,$ $\opt))$ is cut at a level higher than $\log (\dist(c,$ $\opt) \cdot d/\badcut)$ -- not that this is for the analysis only, since we don't know this algorithmically.
We say that a center $f \in L$ is \emph{badly-cut} if for some $i$, the ball $B(f, 2^i)$ is cut at a level higher than $i+ \log (d \log n/\badcutF)$.  
As $L$ and $\calD$ will be fixed all along that section, we simply say that a point or a center is badly-cut.
Notice that we do not know which clients are badly-cut. It is however possible to compute the badly-cut centers, since it depends only on $L$ and $\calD$. It is explained how to perform this step in time $\tilde O(nd)$ in \cite{CohenAddadFS19}.

Our algorithm computes a randomized quadtree $\calD$, and finds the badly-cut center $\calB_\calD$.
It removes from the input each cluster associated with a center of $\calB_\calD$. Let $P_\calD$ be the remaining points.
 $P_\calD$ is a random variable that depends on the randomness of~$\calD$.
Given a solution $S$ for $k$-means on $P_\calD$, the algorithm's output is $S \cup \calB_\calD$.

We call $\cost(P, S)$ the cost of any  solution $S$ in the original input 
$P$, and $\cost(P_\calD, S)$ its cost in $P_\calD$.

A key property for our analysis is a bound on the probability of being badly-cut.
\begin{lemma}\label{lem:badcut}
  Any client $p$ has probability at most $\badcut$ to be badly-cut.
  Similarly, a center $f \in L$ has probability at most $\badcutF$ to be badly-cut.
\end{lemma}
\begin{proof}
  Consider first a point $p \in P$. By \cref{lem:ballCut}, the 
probability that a ball $B(p, r)$ is cut at level at least $j$ is 
at most $d r /2^j$. Hence the probability that a 
ball $B(p,\dist(p, \opt))$
is cut at a level $j$ greater than $\log (\dist(p, \opt)) + \log (d/\badcut)$ is at most $\badcut$.
The proof for $f \in F$ is identical.
\end{proof}

Using that lemma, one can bound the cost of the clusters of facilities from $\calB_\calD$, as well as the cost of badly-cut clients:
\begin{lemma}\label{lem:boundBC}
For any $\pi \in (0, 1)$, it holds with probability $1 - \pi$ that:
\begin{align*}
  \sum_{f \in \calB_\calD} \sum_{p : L(p) = f} \cost(p, f) \leq 3/\pi \cdot \badcutF \cost(P, L), \\
  \sum_{p \bcc} \cost(p, L) \leq 3/\pi \cdot \badcut \cost(P, L)
  \ \text{and} \
  |\calB_\calD| \leq 3/\pi \cdot  \badcutF |L|
  \end{align*}
\end{lemma}
\begin{proof}
Using \cref{lem:badcut}, we have 
$$\E[\sum_{f, p :  L(p) = f} \cost(p, f)] = \sum_{p \in P} \Pr[L(p) \in \calB_\calD] \cost(p, L) \leq \badcutF \cost(P, L).$$
Using Markov's inequality, with probability $1-\pi/3$ the first bullet of the lemma holds. For the same reason, the second bullet holds with probability $1-\pi/3$ as well. Similarly,
$\E[|\calB_\calD|] = \sum_{f \in L} \Pr[f \in \calB_\calD] = \badcutF |L|,$ so applying again Markov's inequality gives that the third bullet holds with probability $1-\pi/3$. A union-bound concludes the proof.
\end{proof}

\begin{lemma}\label{lem:goodInstance}
When $\badcutF = \frac{\pi \alpha}{6}$, $\badcut = \frac{\alpha^3 \pi^3}{144 d^3 \log^2 n}$ and \cref{lem:boundBC} holds:
$$\cost_\calD(P_\calD, \opt) \leq \frac{\alpha}{2} \cost(P, L) +  \frac{O(d^9 \log^4 n)}{\alpha^6 \pi^6} \cdot \cost(P, \opt).$$
\end{lemma}
\begin{proof}
We start by showing different bounds for $\cost_\calD(c, \opt)$, according to whether $c$ is badly-cut or not.

When a client $p$ is not badly-cut, we directly have that:
$\dist_\calD(p, \opt) \leq \frac{d \sqrt d\cdot \dist(p, \opt)}{\badcut}$, since the lowest common ancestor of $p$ and $\opt(p)$ has diameter at most $\frac{d \sqrt d\cdot \dist(p, \opt)}{\badcut}$.{\setlength{\emergencystretch}{2.5em}\par}

In the case where $p \in P_\calD$ is badly-cut, we proceed differently. We use that $L(p)$ is not badly-cut as follows.  
Both $p$ and $\opt(p)$ are contained in the ball $B(L(p), \dist(p, L) + \dist(p, \opt))$, since $\dist(L(p), \opt) \leq \dist(p, L) + \dist(p, \opt)$. 
Let $i = \lceil \log(\dist(p, L) + \dist(p, \opt))\rceil$. Since $L(p)$ is not badly-cut, the ball $B(L(p), 2^{i})$ contains $p$ and $\opt(p)$ and is cut at level at most $i + \log(d \log n/ \badcutF)$. 
Hence, $\dist_\calD(p, \opt) \leq \frac{d^{3/2} \log n}{\badcutF} 2^{i} \leq \frac{2 d^{3/2} \log n}{\badcutF} \cdot \left(\dist(p, L) + \dist(p, \opt)\right)$.{\setlength{\emergencystretch}{2.5em}\par}

Since $\cost_\calD(p, \opt) = \dist_\calD(p, \opt)^2$, this implies that 

$$\cost_\calD(p, \opt) \leq \frac{2 d^3 \log^2 n}{\badcutF^2} \cdot \left(2\cost(p, L) + 2\cost(p, \opt)\right) $$
$$\leq \frac{4 d^3 \log^2 n}{\badcutF^2} \cdot \left(\cost(p, L) + \cost(p, \opt)\right)$$

Hence, we have that:
\begin{flalign*}
    &\cost_\calD(P_\calD, \opt) =\!\!\!\!\!\!\!\!\!\!\! \sum_{p \in P_\calD: p \bcc}\!\!\!\!\!\!\!\!\! \cost_\calD(p, \opt) 
+ \!\!\!\!\!\!\!\!\!\!\!\sum_{p \in P_\calD:  p \text{ not badly-cut}} \!\!\!\!\!\!\!\!\!\!\!\! \cost_\calD(p, \opt)&\\
    &\leq \sum_{p \in P_\calD: \text{ badly-cut}} \frac{4 d^3 \log^2 n}{\badcutF^2} \cdot (\cost(p, \opt) + \cost(p, L)) \\
&\phantom{xxx} + \sum_{p \in P_\calD:  p \text{ not badly-cut}} \frac{d^3 \cost(p, \opt)}{\badcut^2}\\
    \leq& \left(\frac{4 d^3 \log^2 n}{\badcutF^2} + \frac{d^3}{\badcut^2}\right)\cdot \cost(P, \opt) + \!\!\!\!\!\!\!\!\!\!\sum_{p \in P: \text{ badly-cut}} \!\!\!\!\!\!\!\!\!\frac{4 d^3 \log^2 n}{\badcutF^2} \cdot \cost(p, L)
\end{flalign*}

Using now \cref{lem:boundBC}, we get: 
\begin{flalign*}
    &\cost_\calD(P_\calD, \opt \cup \calB_\calD)  \leq \left(\frac{4 d^3 \log^2 n}{\badcutF^2} + \frac{d^3}{\badcut^2}\right)\cdot \cost(P, \opt)+& \\
    &\frac{4 d^3 \log^2 n}{\badcutF^2} \cdot \frac{3\badcut}{\pi} \cost(P, L)
    \leq  \frac{\alpha}{2} \cost(P, L) +  \frac{O(d^9 \log^4) n}{\alpha^6 \pi^6}  \cdot \cost(P, \opt).
\end{flalign*}

\end{proof}

That lemma shows that it is enough to compute the optimal solution on $P_\calD$, and add to it the centers of $\calB_\calD$ which can be done by an algorithm similar to the one for $k$-Median:
\begin{algorithm}
\caption{DP-kMeans($P, L$)}\label{alg:kmeans}
\begin{algorithmic}[1]
\State Compute a shifted quadtree $\calD$. Let $r$ be the root of $\calD$.
\State Compute the instance $P_\calD$, $w = $ MakePrivate($\calD, P_\calD$). 
\State Compute $s = $DynamicProgram-KMeans($\calD, w, k, r$)
\State Use the dynamic program table to find a solution $S$ with cost $s$
\State \textbf{Return} $S \cup \calB_\calD$
\end{algorithmic}
\end{algorithm}

The algorithm DynamicProgram-KMeans is exactly the same as DynamicProgram-KMedian \cref{alg:dynKMed}, except that it returns $w(c) \cdot \diam(c)^2$ at step 3, to fit the $k$-means cost. We can now turn to the proof of \cref{lem:main:kmeans}, to show the guarantees ensured by this algorithm.

\begin{proof}[Proof of \cref{lem:main:kmeans}]
We start by showing the quality of approximation. As for $k$-median, the solution $S$ computed at step 5 of \cref{alg:kmeans} is optimal for $P_\calB$ in the quadtree metric with the additional noise. Hence, its cost verifies $\cost(P_\calD, S) \leq \cost_\calD(P_\calD, S) \leq \cost_\calD(P_\calD, \opt) + \frac{k d^2 \log^2 n}{\eps} \cdot \Lambda^2$. 

Now, with probability $1-\pi$ \cref{lem:boundBC} holds. 
In that case, using \cref{lem:goodInstance}, the cost of $P_\calD$ is at most 
$$\cost(P_\calD, S) \leq \frac{\alpha}{2} \cdot \cost(P, L) + \frac{O(d^9 \log^4 n)}{\alpha^6 \pi^6} \cdot \cost(P, \opt)  + \frac{k d^2 \log^2 n}{\eps} \cdot \Lambda^2.$$

Moreover, \cref{lem:boundBC} ensures that $\sum_{f \in \calB_\calD} \sum_{c : L(c) = f} \cost(c, f) \leq \frac{\alpha}{2} \cost(P, L)$. Hence, combining those bounds concludes the lemma.

We now turn to the privacy guarantee. $\calD$ is computed oblivious to the data. Hence, when $P$ changes by one point, $P_\calD$ changes by at most one point as well -- depending whether this point is served by a badly-cut center in $L$. As for $k$-median, Step 3 of the algorithm therefore ensures that the solution computed at step 5 is $\eps$-DP. 
\end{proof}

\paragraph{Going from $(1+\alpha)k$ centers to $k$}
In this last section, we show how to get a true solution to $k$-means, removing the extra $\alpha k$ centers.

For that, we can use the reverse greedy algorithm of Chrobak et al.~\cite{ChrobakKY06}.
This algorithm starts from a set of centers, and iteratively removes the one leading to the smallest cost increase, until there are $k$ centers remaining.
It can be implemented in a private manner as follows:
let $S$ be a set of $O(k)$ centers computed privately. For any center $s$ of $S$, let $w(s)$ be the size of $S$'s cluster, plus a Laplace noise $\laplacenoise(1/\eps)$.
Let $P_S$ be the resulting instance. Informally, any solution on $P_S$ induces a solution with similar cost on $P$, with an additive error $\pm~\cost(P, S) + k\cdot \Lambda^2 / \eps$ -- see \cref{lem:reverseApprox}. Further, since $S$ is private, $P_S$ is private as well -- see \cref{lem:reversePrivate}.

On the weighted instance $P_S$, the reverse greedy algorithm finds a solution with $k$ centers that is an $O(\log k)$-approximation of the optimal solution $\mathcal A$ for that instance, using that $P_S$ contains $O(k)$ distinct points. This is Theorem 2.2 in \cite{ChrobakKY06}.

Now, the optimal solution  $\calA$ on $P_S$ has cost in $P$ at most $\opt + \cost(P, S) + k \Lambda^2 / \eps$, by \cref{lem:reverseApprox}. Hence, combined with \cref{thm:km}, $\mathcal{A}$ has cost $\cost(P, \mathcal{A}) \leq \poly(d, \log n, 1/\alpha) \cdot \opt + k d^2 \log^2 n \log k/ \eps^2 \Lambda^2$. The solution computed by the reversed greedy has therefore cost at most $\poly(d, \log n, 1/\alpha) \cdot \opt + k d^2 \log^2 n \log^2 k/ \eps^2 \Lambda^2$. {\setlength{\emergencystretch}{2.5em}\par}

Before formalizing the argument, we show the two crucial lemmas.
\begin{lemma}\label{lem:reversePrivate}
If solution $S$ is computed via an $\eps$-DP algorithm, then the algorithm computing $P_S$ is $2\eps$-DP.
\end{lemma}
\begin{proof}
Fix some solution $S$. By properties of the Laplace mechanism (see \ref{sec:prelim}), for any center $s$ of $S$ the value of $w(s)$ is computed  on $s$'s cluster is $\eps$-DP. Since all clusters are disjoint, the instance $P_S$ is computed in an $\eps$-DP way. 

Now, $S$ is not fixed but given privately to the algorithm. By composition, the algorithm that computes $P_S$ is $2\eps$-DP.
\end{proof}

\begin{lemma}\label{lem:reverseApprox}
Let $S$ be any solution, and $P_S$ computed as described previously. With high probability, for any set of $k$ centers $T$, $|\cost(P, T) - \cost(P_S,T)| \leq \frac{1}{2} \cdot \cost(P_S, T) + 10 \cost(P, S) + k \log n \Lambda^2 / \eps$. 
\end{lemma}
\begin{proof}
Without the addition of a Laplace noise, the cost difference between the two solution can be bounded using the following generalization of the triangle inequality (see Lemma 1 in Cohen-Addad et al.~\cite{stoc21}): for any $\eps >0$, any points $p,s$ and set $T$, 
\[|\cost(p, T) - \cost(s, T)| \leq \eps \cost(s, T) + \frac{4+\eps}{\eps} \cost(p, s).\]

For any point $p \in P$ served by some center $s \in S$, we can apply this inequality with $\eps = 1/2$ to get:
\[|\cost(p, T) - \cost(s, T)| \leq\frac{1}{2}\cdot \cost(s, T) + 10 \cost(p, s).\]

Moreover, w.h.p the total noise added is smaller than $k \log n / \eps$, hence contributes at most $k \log n \Lambda^2 / \eps$ to the cost. Summing over all $p$ concludes the proof.
\end{proof}

\begin{lemma}
Let $\calS$ be the solution computed by \cref{thm:km}, and $P_\calS$ the instance computed as described previously.
Applying the reverse greedy algorithm on instance $P_\calS$ is $2\eps$-DP and yields a solution with cost at most $\poly(d, \log n, 1/\alpha) \cdot \opt + k d^2 \log^2 n \log^2 k/ \eps^2 \Lambda^2$ 
\end{lemma}
\begin{proof}
First, the algorithm is $2\eps$-DP, as shown by \cref{lem:reversePrivate}.

Second, let $\calA$ be the optimal solution on the instance $P_\calS$, and $\opt$ be the optimal cost for the full set of points $P$. Applying \cref{lem:reverseApprox}, we get:
\begin{align*}
\cost(P_\calS, \calA) &\leq \cost(P_\calS, \opt) \\
&\leq 2 \cost(P, \opt) + 20 \cost(P, \calS) + 2k \log n \Lambda^2 / \eps\\
&= \poly(d, \log n, 1/\alpha) \cdot \opt + k d^2 \log^2 n \log k/ \eps^2 \Lambda^2.
\end{align*}

We can now bound the cost of the solution computed by the reverse greedy algorithm. As $P_\calS$ is made of $O(k)$ many distinct points, the reverse greedy computes a solution $\tilde \calS$ with cost at most $\cost(P_\calS, \tilde \calS) = O(\log k) \cost(P_S, \calA)$.
Applying again \cref{lem:reverseApprox}, we get
\begin{align*}
\cost(P, \tilde \calS) &\leq \frac{3}{2}\cdot \cost(P_S, \tilde \calS) + 10 \cost(P, \calS) + k \log n \Lambda^2 / \eps\\
&= \poly(d, \log n, 1/\alpha) \cdot \opt + k d^2 \log^2 n \log k/ \eps^2 \Lambda^2,
\end{align*}
which concludes the proof.
\end{proof}

\end{document}